\documentclass[a4paper,11pt,reqno]{amsart}

\usepackage[utf8]{inputenc}
\usepackage{mathrsfs,dsfont,hyperref,amsmath,amssymb,amsthm,amsfonts,amstext,amsopn,amsxtra,mathrsfs,esint,enumitem,bookmark,mathtools}
\usepackage[capitalize]{cleveref}
\usepackage{braket}

\newtheorem{theorem}{Theorem}%[section]

\newtheorem{lemma}[theorem]{Lemma}
\newtheorem{corollary}[theorem]{Corollary}

\theoremstyle{definition}

\theoremstyle{remark}

\newtheorem{remark}[theorem]{Remark}

\newcommand{\F}{\mathcal{F}}
\newcommand{\E}{\mathcal{E}}
\newcommand{\ii}{{\rm i}}

\newcommand{\N}{\mathcal{N}}

\newcommand{\R}{\mathbb{R}}

\newcommand{\Z}{\mathbb{Z}}
\newcommand\1{{\ensuremath {\mathds 1} }}
\newcommand{\dt}{{\rm d}t}
\newcommand{\ds}{{\rm d}s}

\newcommand{\B}{{\mathcal{B}}}

\newcommand{\vphi}{{\varphi}}

\newcommand{\ca}{{\check a}}

\begin{document}

\title[BEC grand canonically]{Another proof of BEC in the GP-limit}

%\author[R. L. Frank]{Rupert L. Frank}
%\address[R. L. Frank]{Mathematisches Institut der Universit\"at M\"unchen, Theresienstr. 39, 80333 M\"unchen, Germany, and Mathematics 253-37, Caltech, Pasadena, CA 91125, USA}
%\email{rlfrank@caltech.edu}

\author[C. Hainzl]{Christian Hainzl}
\address[C. Hainzl]{Mathematisches Institut der Universit\"at M\"unchen, Theresienstr. 39, 80333 M\"unchen, Germany}
\email{hainzl@math.lmu.de}

%\author[R. Seiringer]{Robert Seiringer}
%\address[R. Seiringer]{Institute of Science and Technology, Klosterneuburg, Austria}
%\email{langmann@kth.se}

\begin{abstract}
We present a fresh look at the methods introduced by Boccato, Brennecke, Cenatiempo, and Schlein concerning the trapped Bose gas
and give a conceptually very simple and concise proof of BEC in the Gross-Pitaevskii limit for small interaction potentials. 
\end{abstract}

\maketitle

\section{Introduction}

One of the major achievements in mathematical quantum mechanics within the last 25 years was the proof of Bose-Einstein condensation (BEC) for trapped 
Bose gases  by Lieb and Seiringer in 2002 \cite{LS1}, see also \cite{LSSY,LS2}. Their work was based on preceding works of Dyson \cite{Dy} and Lieb and Yngvason \cite{LY} on the ground state energy of dilute Bose gases. Since then several different proofs of BEC in the Gross-Pitaevskii regime, and beyond, have been carried out. Nam, Rougerie and Seiringer \cite{NRS} gave a proof using the quantum de Finetti theorem. Nam, Napiorkowski, Ricaud and Triaud \cite{NNRT} used ideas of  \cite{BSo,BFS}, while more recently Fournais \cite{F}, by means of techniques developed in his joint work with Solovej \cite{FS} on the Lee-Huang-Yang conjecture, gave a relatively short proof, valid well beyond the GP-regime. While the above approaches are based on localization techniques in configuration space, 
Boccato, Brennecke, Cenatiempo and Schlein (BBCS), prior to \cite{NNRT,F,BSo,FS}, developed a different approach, more in the spirit of Bogolubov's original work. They use  unitary rotations to encode the expected ground state. On the one hand they achieved optimal error bounds in their proof of BEC \cite{BBCS0,BBCS1}, on the other hand  this approach culminated in the rigorous establishment of Bogolubov theory on a periodic box \cite{BBCS}.  
Although their works are voluminous, the methods are conceptually quite accessible and rather straight forward.  

The aim of the present work is to take a fresh look at the approach of BBCS, with the obvious difference 
that we treat the system in a grand-canonical way, inspired by Brietzke and Solovej \cite{BSo}. 
This allows us to conceptually simplify the approach of BBCS and additionally streamline the error estimates. 
A further advantage of this approach is that the emergence of the scattering length in the final result comes out automatically and does 
not have to be put in from the beginning. Thanks to the smallness assumption on the interaction potential we achieve to present a concise proof of BEC in the GP-limit with optimal error bounds. 
It should also be remembered that another advantage of the use of unitary rotations is the fact that one simultaneously produces precise upper and lower bounds. 
On the downside it is fair to mention that one needs regularity assumptions on the interaction $V$, excluding the hard-core potential which is included in the results \cite{LS1,NRS,F,NNRT}. 
The smallness assumption of the interaction potential simplifies our approach significantly. In \cite{ABS} Adhikari, Brennecke and Schlein managed to overcome the smallness condition by an independent argument, without relying on previous results \cite{LS1,NRS}, using additional cubic and quartic transformations, which makes the proof technically even harder than \cite{BBCS0,BBCS1}. This suggests that getting rid of this smallness condition, however, seems to be a major task within this approach.

We consider a system of bosons in the Gross-Pitaevskii regime by means of the 
grand canonical Hamiltonian 
$$ H_\mu = H_N - \mu \N ,$$
with $H_N = \bigoplus_{n=0}^\infty H_N^n$, where the $n$-particle Hamiltonian is given by
$$
H_N^n = \sum_{i=1}^n -\Delta_i + \kappa \sum_{i<j}^n V_N(x_i-x_j).
$$
$H_\mu$ is 
acting on the bosonic Fock space 
$$\mathcal{F}(\mathcal{H}) = \bigoplus_{n=0}^\infty \mathcal{H}^{\otimes_s n}$$
with  $\mathcal{H} = L^2(\Lambda)$, with $\Lambda= [-1/2, 1/2]^3$. 
Further the GP-regime is reflected by the scaling of the potential 
%We assume the particles being confined in a box $\Lambda = [-1/2,1/2]^3$ interacting through a positive $V(x)$, which, for convenience, is assumed to 
%have compact support. 
%with scattering length of the order of $1/N$. 
%We further impose periodic boundary conditions. The 
 $$V_N(x) :=  N^2 V(Nx),$$where $V(x)$ is assumed to be positive and compactly supported and $\in L^1(\R^3) \cap L^3(\R^3)$. We impose periodic boundary conditions on the
box $\Lambda= [-1/2, 1/2]^3$. In that sense the Hamiltoinan $H_\mu$  should actually contain the periodized potential. However, we will work  mainly with the variant in momentum space, where the periodization is automatic.

%Let us observe that $N$ acts as a parameter, however, we choose the chemical potential, similar to \cite{BSo}, in such a way that the expected number of particles in the ground state, to leading order, is $N$. More precisely, we choose the chemical potential as $\mu = 8 \pi a$, where $a$ is the scattering length of $\kappa V$.

Notice that $N$ in the Hamiltonian $H_\mu$ acts as a parameter. However, we choose the chemical potential so that the expected number of particles in the ground state, to the leading order, is $N$. More precisely, we follow \cite{BSo} and choose the chemical potential as $\mu = 8 \pi a$, where $a$ is the scattering length of $\kappa V$. 

In contrast to previous works we define the scattering length via its Born series, in the form
\begin{equation}\label{sclBorn}
4 \pi a := \frac \kappa 2 \int_{\R^3} V(x) dx - \left \langle \frac \kappa 2 V, \frac{1} {-\Delta + \frac \kappa 2 V} \frac \kappa 2 V \right \rangle,
\end{equation}
for the simple reason that this is exactly the way how the scattering length appears in our approach. 
In a more concise way, see \cite{HS}, the scattering length \eqref{sclBorn} can also be expressed as
$$ 4 \pi a = \left \langle \sqrt{v} , \frac 1 {1 + \sqrt{v} \frac 1{-\Delta} \sqrt{v}} \sqrt{v} \right \rangle,
$$
with $v = \frac \kappa 2 V$.

\section{Main results} 

For convenience we rewrite the Hamiltonian $H_\mu$ in momentum space via 
$$
H_\mu = \sum_p (p^2 - \mu)  a_p^\dagger a_p + \frac{\kappa}{2}\sum_{r,p,q} \hat{V}_N (r) a_{p+r}^\dagger a_q^\dagger a_p a_{q+r},
$$
with $p \in \Lambda^\ast :=  2\pi \mathbb{Z}^3$, and $a^\dagger_p := a^\dagger(\phi_p)$, $a_p :=  a(\phi_p)$
the usual creation and annihilation operators on the Fock space over the  periodic box, i.e., $\phi_p(x) = e^{ix\cdot p}$. 
Notice, $$\hat{V}_N (r) = \frac 1N \hat V(r/N).$$
%
%Let $\phi_p(x) = e^{ixp}, p \in \Lambda^\ast :=  2\pi \mathbb{Z}^3$. This gives an ONB of $L^2(\Lambda)$. 
%The $\phi_p$ are eigenfunctions with periodic boundary conditions of $-\Delta$ with eigenvalues $p^2$. \\ 
%Let $a^\dagger_p := a^\dagger(\phi_p)$ and $a_p :=  a(\phi_p)$. From \eqref{CCR} we immediately see 
%\[ \label{CCRdiscrete}
%[a^\dagger_p, a^\dagger_q] = 0 = [a_p, a_q], \quad [a_p, a^\dagger_q] = \delta_{p,q}.
%\]
%

Let 
$$E_\mu(N) := \inf\{\langle \psi, H_\mu\psi \rangle | \psi \in \mathcal{F}, \|\psi\| = 1 \} $$
be the ground stat energy on the Fock space. The first theorem concerns the grand-canonical ground state energy. 
The statement resembles well known results in the literature, e.g. \cite{LSY,BBCS0,BBCS}. 
Its proof forms the basis for the subsequent establishment of BEC. 
\begin{theorem} \label{energyBound} 
Let $\mu = 8 \pi a$. Then for $\kappa$ small enough 
$$
E_\mu(N) = - 4 \pi a N + O(1) ,
$$
asymptotically as $N$ tends to infinity. 
\end{theorem}
In the proof of Theorem \ref{energyBound} we use the fact, that the number operator $\N$ commutates with the Hamiltonian, i.e. $[H_\mu, \N]=0$, which allows us to restrict the determination of the ground state energy as well as the proof of BEC to eigenfunctions $\psi$ of the number operator, $\N\psi = n\psi$. 
\begin{remark}
As a consequence of the proof of Theorem \ref{energyBound} we see that any approximate ground state $\psi_n $, with fixed particle number $n$, $\N \psi_n = n \psi_n$,
whose energy is $O(1)$ away from $E_\mu(N)$, i.e.,  $  \langle \psi_n, H_\mu \psi_n \rangle \leq - 4 \pi a N + O(1),$ satisfies 
$ |n/N - 1| \leq  O(1/\sqrt{N}),$ meaning that a-priori $n$ equals the external parameter $N$ only up to an error of $\sqrt{N}$. 
\end{remark}
%In the proof of Theorem \ref{energyBound} we use that the number operator commutes with the Hamiltonian, i.e., $[H_\mu, \N]=0$, allowing us to restrict to eigenfunctions of  the number operator, which simplifies the analysis, and the proof of BEC in the ground state of $H_\mu$. 
\begin{theorem}\label{BEC} 
Let $\mu = 8\pi a$, and let $\psi \in \F$ be  normalized, with $\N \psi = n \psi$ and satisfy $$ \langle \psi, H_\mu \psi \rangle \leq - 4 \pi a N + O(1).$$
Then,  for $\kappa$ small enough,
\begin{equation}
\langle \psi, \N_+ \psi \rangle  \leq O(1).
\end{equation}
Since,
$$
\langle\phi_0, \gamma_\psi  \phi_0 \rangle = \langle \psi, a_0^\dagger a_0 \psi \rangle = \langle \psi, (\N - \N_+) \psi \rangle = n - \langle \psi, \N_+ \psi \rangle = n + O(1),$$
this implies complete BEC. 
\end{theorem}

More precisely, the highest eigenvalue of the one particle density matrix $\gamma_\psi$ of any approximate ground state $\psi$, with $\N\psi=n\psi$
is macroscopically occupied, 
implying BEC with optimal error bounds in terms of $n$. Let us recall that in terms of creation and annihilation operators the one particle density matrix $\gamma_\psi$ can be expressed via the matrix elements $\langle \psi, a^\dagger_p a_q \psi \rangle$, with $p,q \in 2\pi \Z^3$, i.e., $\hat \gamma_\psi(p,q) = \langle \psi, a^\dagger_p a_q \psi \rangle.$
\begin{remark}\label{psiN}
As easy consequence of the proof  of Theorem \ref{energyBound} we will see that one can find a state $\Psi_N$,  with $\N \psi_N = N \psi_N$, and 
$$\langle \psi_N, H_\mu \psi_N\rangle \leq -4 \pi a N + O(1) ,$$
such that the corresponding one-particle density matrix $\gamma_{\psi_N}$ satisfies 
$$\langle \phi_0, \gamma_{\psi_N} \phi_0 \rangle  = N + O(1),$$
with optimal rate in the parameter $N$. 
\end{remark}

As corollary of Theorem \ref{energyBound} and Remark \ref{psiN} we immediately obtain the energy asymptotic of the ground state energy of an $N$-particle system.
Recall, $$
H_N^N = \sum_{i=1}^N -\Delta_i + \kappa \sum_{i<j}^N V_N(x_i-x_j).
$$
\begin{corollary}\label{EN}
Let $E_N = \inf {\rm{spec} } \, H^N_N$, then
$$ E_N = 4 \pi a N + O(1).$$
\end{corollary}
\begin{proof}
The upper bound is provided in the proof of Theorem \ref{energyBound}, cf. Remark \ref{psiN}, by an explicit trial state with particle number $N$. 
The lower bound follows from a simple variational argument \cite{BSo}. Let  $\psi_N$ be the ground state of $H_N^N$. Then, ($\mu = 8\pi a$),
$$ E_N = \langle\psi_N,H_N^N \psi_N\rangle = \langle\psi_N,H_N \psi_N\rangle= 8 \pi a N + \langle \psi_N,H_\mu \psi_N\rangle \geq 8 \pi a N + E_\mu \geq 4\pi a N + O(1),$$
using Theorem \ref{energyBound}. 
\end{proof}

\begin{remark} \label{energyBound2} 
The proof of Theorem \ref{energyBound} can easily be extended to general values of the chemical potential $\mu >0$. 
Indeed, for any $\mu > 0$ and $\kappa$ small enough one gets
$$
E_\mu(N) = - \frac{\mu^2}{16\pi a} N + O(1) ,
$$
as $N$ tends to infinity. 
For approximate ground states  $\psi \in \F$, with $\N \psi = n \psi$, and $$ \langle \psi, H_\mu \psi \rangle \leq - \frac{\mu^2}{16\pi a} N + O(1)$$
one obtains again complete condensation $\langle \psi, \N_+ \psi \rangle  \leq O(1)$, however, the expectation number of particles is now 
$n = \frac{N \mu}{8 \pi a} +O(\sqrt{N})$. 
\end{remark}
In the following section \ref{strategy} we present the main steps of the proof of Theorem \ref{energyBound} and complete the proof of BEC in section \ref{BEC-sec}.
The rest of the paper is concerned with technical estimates which are not important for the understanding of the main ideas of the proof. 

\section{Strategy and main steps of the proofs}\label{strategy}

Notice that the Hamiltonian $H_\mu$ and the number operator $\N$ commute, which tells us that we can restrict to states with 
fixed quantum number $\N \psi = n\psi$. Following ideas from Brietzke and Solovej \cite{BSo} we can restrict our attention to the case where 
$\N \leq 10 N$, with $10$ being chosen for aesthetic reasons (anything larger than $4$ would do). 

The key observation from \cite{BSo} is that whenever there are more than $10N$ particles one can combine them in groups with each group consisting of a number of particles 
between $5N$ and $10N$. Since the interaction is positive, one can simply drop the interaction between different groups for a lower bound. Since we will further show that the energy of a system with more than  $5N$ particles is actually nonnegative, this  tells us 
that we can restrict from the very beginning to $\N \leq 10N$. Notice that in the grand canonical case with positive chemical potential it is easy to see that the ground state must necessarily be negative. 

Under the assumption of $\N \leq 10N$ we are now in the position to apply the strategy developed by Boccato, Brennecke, Cenatiempo, and Schlein \cite{BBCS0,BBCS1,BBCS}, based on ideas of \cite{BS}. 
%We follow the strategy of BBCS and 
%introduce a Bogolubov transformation $e^B$, with  $$B = \frac{1}{2}\sum_{p \neq 0} \vphi_p [a_p^\dagger a_{-p}^\dagger - a_{-p}a_p].$$ 

We will look for an appropriate unitary rotation $e^{\B}$, with $\B = B - B^*$ a number conserving operator on the Fock space, which encodes the ground state, 
 in the sense that
$e^{-\B} H_\mu e^{\B}$ has, to leading order, $\Pi_{i=1}^N\phi_0$ as approximate ground state. This further implies that 
$\psi~\simeq e^{\B} \Pi_{i=1}^N\phi_0$ is an approximate minimizer for $H_\mu$. 

First we follow Bogolubov's way and decompose the interaction potential in different terms depending on the number of $a_0$'s and $a^\dagger_0$'s:
\begin{multline}
\frac{\kappa}2 \sum_{p,q,r} \hat{V}_N(r)a_{p+r}^\dagger a_q^\dagger a_p a_{q+r}
= 
\frac{\kappa}2 \hat{V}_N(0) a^\dagger_0 a^\dagger_0 a_0 a_0 + 
\kappa\sum_{p\neq 0} \hat{V}_N(0) a_p^\dagger a_p a^\dagger_0 a_0  \\ + \kappa \sum_{r\neq 0} \hat{V}_N(r)a_r^\dagger a_r a^\dagger_0 a_0  +  \frac{\kappa}{2} \sum_{r\neq 0} \hat{V}_N(r) \left[a_r^\dagger a_{-r}^\dagger a_0 a_0  + a_{-r}a_r a^\dagger_0 a^\dagger_0 \right] + \\
\kappa \sum_{q,r,q+r\neq 0} \hat{V}_N(r)\left[a_{q+r}^\dagger a_{-r}^\dagger a_{q} a_0  + a_q^\dagger a_{-r} a_{q+r} a_0^\dagger \right] +  \frac{\kappa}{2}\sum_{p,q \neq 0, r \neq -p, r\neq -q} \hat{V}_N(r) a_{p+r}^\dagger a_q^\dagger a_p a_{q+r} 
\end{multline}
Let us denote the number operator counting the number of particles in the state  $\phi_0$ as
$$ \N_0 = a^\dagger_0 a_0=\N - \N_+.$$

Since $$ a^\dagger_0 a^\dagger_0 a_0 a_0 = \N_0(\N_0 - 1) = (\N - \N_+)(\N - \N_+ -1) = \N(\N - 1) - \N_+(2 \N - 1) + \N_+^2,$$ 
we can 
rewrite the Hamiltonian 
$H_\mu$ in the form 
$$ H_\mu = H_0(\mu) + H_1 + H_2 + Q_2 + Q_3 + Q_4,$$
where 
$$H_0(\mu) = \frac{\kappa \hat{V}(0) }{2N}  \mathcal{N} (\mathcal{N} - 1 )  - \mu \mathcal{N}, \qquad H_1 = \sum_{p \neq 0} p^2  a_p^\dagger a_p,$$ 
$$ H_2 = \frac{\kappa}{N} \sum_{p\neq 0} \hat V(p/N) a^\dagger_p a_p(\N - \N_+)  - \frac {\kappa \hat V(0) }{2N} \N_+(\N_+ -1) $$
The rest of the interaction then has the form
\begin{align*}
Q_2 = & \frac{\kappa}{2 N}\sum_{p \neq 0} \hat{V}(p/N) [a_p^\dagger a_{-p}^\dagger a_0 a_0 + h.c.] \\
Q_3 = & \frac {\kappa}{N} \sum_{q,r,q+r\neq 0} \hat{V}(r/N)\left[a_{q+r}^\dagger a_{-r}^\dagger a_{q}a_0 + h.c.  \right]   \\
Q_4 = & \frac{\kappa}{2 N}\sum_{p,q \neq 0, r \neq -p, r\neq -q} \hat{V}(r/N) a_{p+r}^\dagger a_q^\dagger a_p a_{q+r} 
\end{align*}

In the following we will assume that  $a_p$ or $a^\dagger_p$ automatically means that $p\neq 0$ which allows us to skip 
the distinctions in the sums. 
E. g., instead of $\sum_{p \neq 0} a^\dagger_p a_p $ we simply write  $\sum_{p} a^\dagger_p a_p $. 

Let us recall Duhamel's formula
$$
e^{-B}Ae^B = A + \int_0^1 e^{-sB} [A,B]e^{sB}\ds  =A + [A,B] + \int_0^1 \int_0^s e^{-tB} [[A,B],B] e^{tB} \dt \ds.
$$
Applying the second equality of the formula to $H_1 + Q_4$ and the first to $Q_2$  we obtain 
\begin{align}\label{BHm}
e^{-\B}H_\mu e^\B &= e^{-\B} [H_0(\mu) + H_1 + Q_4 + Q_2 + H_2 + Q_3]e^\B 
\\ \nonumber
&= H_0(\mu) + H_1 + Q_4 + [H_1 + Q_4,\B] + Q_2 + e^{-\B}(H_2 + Q_3)e^\B  
\\ \nonumber
& \qquad +  \int_0^1 \int_0^s e^{-t\B}\Big[ [H_1+Q_4,\B],\B \Big] e^{t\B} \dt \ds  + \int_0^1 e^{-s\B}[Q_2,\B]e^{s\B} \ds .
\end{align}
Let us explain the main idea of the strategy. The term $H_0(\mu)$ clearly contributes to the leading term. In Bogolubov's original approach
$Q_3$ and $Q_4$ was omitted and Bogolubov \cite{bog} diagonalized the quadratic part $H_1 + H_2 + Q_2$. But he did not get the leading term correctly, since he missed the contribution coming from $Q_4$. We perform an ``almost'' diagonalization 
by  choosing $\B$ in such a way that $$[H_1 + Q_4,\B] + Q_2 \simeq 0.$$  
We will treat $H_2$ and $Q_3$ as error terms, since they do not contribute to the leading order. 
The requirement that  $[H_1 + Q_4,\B] + Q_2$ vanishes apart from higher order terms, suggests a choice of $\B$, of the form
\begin{equation}\label{defB}
 \B = \frac 1{2N} \sum_p \vphi_p [ a^\dagger_p a^\dagger_{-p} a_0 a_0 - \rm{h.c.}] = B - B^*,
\end{equation}
with  $\vphi_p$ appropriately chosen. In fact, to leading order,  $\vphi_p$ will satisfy the scattering equation.
\begin{lemma}
Let $\B$ be defined as in \eqref{defB}.
If $\vphi_p$ satisfies the equation 
\begin{equation} \label{scattequ}
p^2 \vphi_p + \frac{\kappa}{2N}\sum_q \hat{V}\left((p-q)/N\right)  \vphi_q = -\frac{\kappa}{2} \hat{V}(p/N),
\end{equation} 
then 
\begin{equation}\label{realscattequ}
[H_1(\mu) + Q_4,\B] = -Q_2 + \Gamma,
\end{equation}
with 
\begin{align*}
\Gamma &= \frac \kappa {N^2}  \sum_{r,p,q} \hat{V}(r/N) \vphi_p[a_{p+r}^\dagger a_q^\dagger a_{-p}^\dagger a_{q+r}a_0a_0 + h.c.] \end{align*}
\end{lemma}
\begin{proof}
Straightforward calculations yield
$$
[H_1,\B] = \frac 1N \sum_{p } p^2 \vphi_p [a_p^\dagger a_{-p}^\dagger a_0 a_0  + h.c.]
$$
and 
\begin{multline}
[Q_4, \B] = \frac{\kappa}{2N^2} \sum_{p,q} (\hat{V}((p-q)/N) \vphi_q a_{p}^\dagger a_{-p}^\dagger a_0a_0 
\\+ \frac{\kappa}{2N^2} \sum_{r,p,q} \hat{V}(r/N) \left(\vphi_pa_{p+r}^\dagger a_q^\dagger a_{-p}^\dagger a_{q+r}a_0 a_0 +  
 \vphi_{q+r} a_{p+r}^\dagger a_q^\dagger a_{-q-r}^\dagger a_p a_0 a_0 \right) + h.c.,
\end{multline}
where the two terms in the last line are actually equal, which can be seen by changing variables, $p \to p-r$, $q\to q-r$ and then $r \to -r$.  
Collecting all terms involving  $a_p^\dagger a_{-p}^\dagger a_0 a_0$ and recalling the form of $Q_2$ 
we see that  
$$[H_1 + Q_4,\B] = -Q_2 + \Gamma$$ is satisfied if 
 $\vphi_p$ solves the equation \eqref{scattequ}. 
\end{proof}
\begin{remark}
Let us apply the 
discrete inverse Fourier transform 
\begin{equation}\label{invfour}
P_0^\perp \check \vphi(x) = \frac 1{N^3} \sum_p e^{\ii \frac pN \cdot x } \vphi_p, \quad  P_0^\perp \check {\hat V}(x) = \frac 1{N^3} \sum_p e^{\ii \frac pN \cdot x } \hat V(p/N), 
\end{equation}
where $P_0$ is the projection on the constant function $\phi_0$, and the orthogonal projection comes about because all sums run over $p \neq 0$. 
Applying this transformation to 
equation \eqref{scattequ} and assume that $N$ is large enough that $\check {\hat V} (x) = V(x)$ on $[-N/2,N/2]$, then we obtain the equation
\begin{equation}\label{scattequ2}
P_0^\perp (-\Delta + \frac \kappa 2 V(x) ) P_0^\perp \check \vphi(x) = -\frac{\kappa}{2N^2} P_0^\perp V(x),
\end{equation} 
which can be inverted by 
\begin{equation}\label{deffi}
P_0^\perp  \check \vphi(x) = - \frac{\kappa}{2N^2}  \frac 1{P_0^\perp(-\Delta + \frac \kappa 2 V)P_0^\perp } P_0^\perp V (x).
\end{equation}
Among others this shows that equation \eqref{scattequ} has a unique solution $\vphi_p$. Useful properties of this function $\vphi_p$ are provided
in Lemma \ref{Lemmavphi}. 
\end{remark}

We continue with equation \eqref{BHm}. We plug  $[H_1+ Q_4,\B] = -Q_2 + \Gamma$ into the last two terms in \eqref{BHm} and obtain
\begin{align*}
&\int_0^1 \int_0^s e^{-t\B}\Big[ [H_1+Q_4,\B],\B \Big] e^{t\B} \dt \ds + \int_0^1 e^{-t\B}[Q_2,\B]e^{t\B} \dt
\\
&= - \int_0^1 \int_0^s e^{-t\B}[Q_2,\B] e^{t\B} \dt \ds + \int_0^1 \int_0^1 e^{-t\B}[Q_2,\B]e^{t\B} \dt \ds  + \int_0^1 \int_0^s e^{-t\B}[\Gamma,\B] e^{t\B} \dt\ds
\\
&=  \int_0^1 \int_s^1 e^{-t\B}[Q_2,\B]e^{t\B} \dt \ds +  \int_0^1 \int_0^s e^{-t\B}[\Gamma,\B] e^{t\B} \dt\ds.
\end{align*}
Rewriting \eqref{BHm} accordingly we arrive at
\begin{align} \label{expH_contracted}\nonumber
e^{-\B}H_\mu e^\B &= H_0(\mu)+\int_0^1 \int_s^1 e^{-t\B}[Q_2,\B]e^{t\B} \dt \ds + H_1 + Q_4 \\ &+  \Gamma + \int_0^1 \int_0^s e^{-t\B}[\Gamma,B] e^{t\B} \dt\ds 
 + e^{-\B}(H_2 + Q_3)e^\B.
\end{align}
The idea now is rather simple. The leading order is contained in the first two terms on the right hand side, $$H_0(\mu)+\int_0^1 \int_s^1 e^{-t\B}[Q_2,\B]e^{t\B} \dt \ds.$$
The positive term $H_1 + Q_4$ is used to dominate the error terms coming from the rest. Here the smallness of the potential $\kappa$ will be used. Thanks to the gap, 
all errors of the form $\kappa C \N_+$ will be absorbed by $H_1$ for small enough $\kappa$. 
In the following we extract the leading contribution of the term $\int_0^1 \int_s^1 e^{-t\B}[Q_2,B]e^{t\B} \dt \ds.$
\begin{lemma}One has 
\begin{equation}\label{CQ2B}
 \int_0^1 \int_s^1 e^{-t\B}[Q_2,B]e^{t\B} \dt \ds = \frac {\kappa  \mathcal{N}(\N-1)}{2N^2} \sum_p \hat{V}(p/N) \vphi_p + \Xi,
 \end{equation}
 with 
\begin{multline}
 \Xi = -\frac{\kappa} {N^2} \sum_p  \hat{V}(p/N) \vphi_p   \int_0^1 \int_s^1 e^{-t\B} \N_+(2 \N - \N_+ -1 )e^{t\B}\dt \ds \\ +  
 \frac{2\kappa} {N^2}  \sum_p \hat{V}(p/N) \vphi_p \int_0^1 \int_s^1 e^{-t\B}\N_0(\N_0 -1) a^\dagger_p a_pe^{t\B} \dt \ds  \\
 -   \frac{\kappa}{N^2} \sum_{p,q}  \hat{V}(p/N) \vphi_q \int_0^1 \int_s^1 e^{-t\B}  a^\dagger_{p} a^\dagger_{-p} a_{-q} a_{q} (1 + 2 \N_0)  e^{t\B} \dt \ds.
\end{multline}
\end{lemma}
\begin{proof}
We  calculate
\begin{align}\nonumber
[Q_2,B] &= \frac{\kappa}{4N^2} \sum_{p,q} \hat{V}(p/N) \vphi_q [a_p^\dagger a_{-p}^\dagger a_0a_0 + a_{-p}a_pa_0^\dagger a_0^\dagger , a_q^\dagger a_{-q}^\dagger a_0a_0 - a_{-q}a_q a_0^\dagger a_0^\dagger] 
\\ \label{Q2B}
%&= \kappa \frac{\mathcal{N}_0}{4} \sum_{p,q} \hat{V}_N(p) \vphi_p \left([a_{-p} a_p,a_q^\dagger a_{-q}^\dagger] + [a_{-q} a_q,a_p^\dagger a_{-p}^\dagger]\right) 
%\\
=&  \frac{\kappa}{4N^2} \sum_{p,q} \hat{V}(p/N) \vphi_q \left( [a_{-p}a_pa_0^\dagger a_0^\dagger , a_q^\dagger a_{-q}^\dagger a_0a_0] - [a_p^\dagger a_{-p}^\dagger a_0a_0,a_{-q}a_q a_0^\dagger a_0^\dagger] \right)
\end{align}
The two terms in the bracket are hermitian conjugates. Hence it suffices to calculate the second one, 
\begin{align*}
  - [a_p^\dagger a_{-p}^\dagger a_0a_0,a_{-q}a_q a_0^\dagger a_0^\dagger] &= [a_{-q}a_q, a_p^\dagger a_{-p}^\dagger] a_0^\dagger a_0^\dagger a_0a_0 - a_p^\dagger a_{-p}^\dagger a_{-q} a_q [a_0a_0,a_0^\dagger a_0^\dagger], 
 \end{align*}
where
 $$ [a_{-p}a_p,a_q^\dagger a_{-q}^\dagger ] = (\delta_{p,q} + \delta_{p,-q})( 1 + a^\dagger_p a_p + a^\dagger_{-p} a_{-p} ),$$
 and
  $$a_0^\dagger a_0^\dagger a_0 a_0 = \N_0(\N_0-1)= \N(\N -1) - 2 \N \N_+ + \N_+(\N_+ -1), \quad [a_0a_0,a^\dagger_0a^\dagger_0] = 2(2 \N_0 + 1).$$
 Plugging into \eqref{Q2B} and integrating over $s,t$ implies the statement. Observe that for the leading term the integral over $s,t$ gives a factor $1/2$. 
\end{proof}
Let us now define 
\begin{equation}\label{defaN} 
4 \pi a_N := \frac \kappa 2 \left( \hat V(0) + \frac 1N \sum_p \hat{V}(p/N) \vphi_p \right). 
\end{equation}
With this definition and the previous Lemma we rewrite \eqref{expH_contracted} as
\begin{equation} \label{expH_error_decomp}
e^{-\B} H_\mu e^\B = 4 \pi a_N \frac{\N(\N-1) }{N} - \mu \mathcal{N}  + H_1+ Q_4 +\Xi +  \mathcal{E}
\end{equation}
where 
\begin{align*}
\mathcal{E} &=  
\Gamma 
 + \int_0^1\int_0^s e^{-t\B}[\Gamma,\B] e^{t\B}\dt \ds + e^{-\B}(H_2 + Q_3)e^\B.
\end{align*}
Let us recall that we are able to restrict to wavefunctions with fixed particle number, 
$\N \psi = n \psi$, (and additionally assume $n \leq 10N$). 
Further, 
\begin{multline}  \label{equ:1}
\langle \psi , H_\mu \psi\rangle = \langle e^{-\B} \psi , e^{-\B} H_\mu e^{\B} e^{-\B}  \psi\rangle =\\  =  4 \pi a_N \frac{n(n-1) }{N} - \mu n  
 +\langle e^{-\B} \psi , ( H_1+ Q_4 +\Xi +  \mathcal{E}) e^{-\B}  \psi\rangle,
\end{multline} 
using that $$\N e^{-\B} \psi = e^{-\B} \N \psi = n e^{-\B} \psi.$$
The following lemma, which was proven in \cite{BBCS0}, tells us that the error terms $\E + \Xi $ can be absorbed by  $H_1 + Q_4$. 
\begin{lemma}[\cite{BBCS0}]\label{error}
Assume $\psi \in \F$ with  $\N \psi =n \psi, $ with $ n \leq 10N$ and denote $\xi = e^{-\B} \psi$. Then for $\kappa$ small enough 
\begin{equation}\label{mE}
| \langle \xi, (\E + \Xi ) \xi \rangle | \leq \frac 12 \langle \xi, (H_1 +  Q_4)  \xi \rangle. 
\end{equation} 
\end{lemma}
Furthermore, we now point out that $a_N$ converges to the scattering length, see \cite{BS, BBCS0}. 
\begin{lemma}\label{scat}
Let $a$ be the scattering length defined in \eqref{sclBorn}. Then we obtain for $a_N$, defined in \eqref{defaN}, 
\begin{equation}
|a_N - a| \lesssim 1/N. 
\end{equation}
\end{lemma}
We postpone the proof of Lemma \ref{scat}. 
Applying these two lemmata to \eqref{equ:1} we conclude that there is a constant $C$, such that 
\begin{align}\nonumber 
\langle \psi, H_\mu \psi \rangle & \geq 4 \pi a \frac {n^2 }{N} - \mu n - C + \frac 12 \langle \xi ,H_1\xi \rangle \\ \label{end}
& = 4 \pi a N \left[ \frac n N - \frac{\mu}{8 \pi a } \right]^2 - 4 \pi a N \left( \frac{\mu}{8 \pi a } \right)^2 - C + \frac 12 \langle \xi ,H_1\xi \rangle. 
\end{align}

\begin{lemma} \label{Nbigger}
With $\mu = 8\pi a$  and $\N \psi = n \psi$, and $n \in [5N,10N]$, then 
\begin{equation}
 \langle \psi, H_\mu \psi \rangle \geq 0
\end{equation}
for $N$ large enough. 
\end{lemma} 
\begin{proof}
Equation \eqref{end}  implies 
$$   \langle \psi, H_\mu \psi \rangle \geq 4\pi a N\left [ ( (n/N) - 1)^2  - 1 \right] - C  \geq 4 \pi a N (4^2 -1) - C > 0,$$
for $N$ sufficiently large. 
\end{proof}
Lemma \ref{Nbigger} and equation \eqref{end} imply for $\mu = 8\pi a$ that 
$$\langle \psi, H_\mu \psi \rangle  \geq - 4 \pi a N - C + \frac 12 \langle \xi ,H_1\xi \rangle.$$
This implies the lower bound in the 
statement of Theorem \ref{energyBound}. 

The upper bound is obtained 
using the simple trial state $\tilde \psi = e^{\B} \Pi_{i=1}^N \phi_0$ plugged into \eqref{equ:1}.
This implies 
$$ E_\mu(N) \leq \langle \tilde \psi, H_\mu \tilde \psi \rangle = \langle \Pi_{i=1}^N \phi_0, e^{-\B} H_\mu e^{\B}  \Pi_{i=1}^N \phi_0 \rangle = 4 \pi a_N \frac{N(N-1)}N - 8\pi a N = - 4 \pi a N + O(1) ,$$
where  we used the simple fact $$\langle \Pi_{i=1}^N \phi_0, (H_1+ Q_4 +\Xi +  \mathcal{E}) \Pi_{i=1}^N \phi_0 \rangle =0.$$

\section{Proof of BEC} \label{BEC-sec}

Let us restrict to states $\psi \in \F$, with $\N\psi = n\psi$, which are approximate ground states, i.e., 
$$\langle \psi, H_\mu  \psi \rangle \leq - 4 \pi a N + O(1).$$
 Equation  \eqref{end}  implies for such $\psi$ and $\mu = 8\pi a$ that
\begin{equation}\label{cond}
- 4 \pi a N + O(1) \geq \langle \psi, H_\mu  \psi \rangle \geq 4 \pi a N \left[ \frac n N - 1 \right]^2 - 4 \pi a N - C + \frac 12 \langle \psi, e^{\B} \N_+ e^{-\B} \psi \rangle,
\end{equation}
using $H_1 \geq \N_+$. This shows on the one hand  that  $$\langle \psi, e^{\B} \N_+ e^{-\B} \psi \rangle \leq O(1),$$ 
and on the other hand we obtain 
$$ \left(\frac nN -1\right)^2 \lesssim \frac 1N,$$ which yields
$$ n = N + O(\sqrt{N}).$$ 
Hence, in order to deduce condensation it suffices to show that $\N_+$ is invariant under unitary transformation of  $ e^{\B}$, 
at least for $\N \leq 10N$. 
\begin{lemma}[\cite{BS}] \label{N_+selfbound}
There is a constant $C > 0$ such that for all $t \in [-1,1]$, as operator on $\chi_{\N \leq 10 N} \F$,
\begin{equation}\label{N+1}
e^{-tB}(\mathcal{N}_+ + 1) e^{tB} \leq C(\mathcal{N}_+ +1).
\end{equation}
More general one has for $n = \{1,2,3\}$ ,
\begin{equation}\label{N+2}
e^{-tB}(\mathcal{N}_+ +1)^n e^{tB} \leq C_n (\mathcal{N}_+ +1)^n.
\end{equation}
\end{lemma}
\begin{proof}
Denote 
$$ F(t) = e^{-t\B} (\N_+ + 1)e^{t\B}.$$
Therefore, via explicit calculation,
\begin{multline}
\frac d{dt} F(t) = e^{-t\B} [\N_+, \B] e^{t\B}= \frac 1{2N} e^{-t\B} \sum_{p,k} \vphi_k [ a^\dagger_p a_p, a^\dagger_k a^\dagger_{-k} a_0 a_0 - \rm{h.c.}]e^{t\B} \\
= \frac 1{2N} e^{-t\B} \sum_{k} \vphi_k (a^\dagger_k a_0 a^\dagger_{-k} a_0 + \rm{h.c.})e^{t\B}  \\
\leq \frac 1{N} e^{-t\B} \sum_{k} (a^\dagger_k a_0 a_k  a_0^\dagger +  \vphi_k^2 a_k a^\dagger_k a_0^\dagger a_0)e^{t\B} \\
= \frac 1{N} e^{-t\B} \sum_{k} (a^\dagger_k a_k  (\N_0 (1 + \vphi_k^2) + 1)  +  \vphi_k^2 \N_0) e^{t\B} 
\lesssim  F(t),
\end{multline}
where we used that $\N_0 \leq \N \leq 10 N$, and the fact that $\vphi_k$ has bounded infinity- and $2$-norm, see \eqref{infvphi}, \eqref{2vphi}
With Groenwall's Lemma we obtain \eqref{N+1}.

With respect to \eqref{N+2} lets look at the case $n=3$. 
Then 
\begin{multline}
\frac d {ds} e^{-s\B} ( \mathcal{N}_+ + 1)^3 e^{s\B} = e^{-s\B} [( \mathcal{N}_+ + 1)^3,B] e^{s\B}
\\ = e^{-s\B} \left( (2 \mathcal{N}_+ + 1)^2[\N_+,\B] +( \mathcal{N}_+ + 1) [\N_+,\B] ( \mathcal{N}_+ + 1)  + [\N_+,\B]  ( \mathcal{N}_+ + 1)^2 \right) e^{s\B}
\end{multline} 
Using further that $\N_+ a^\dagger_p = a_p^\dagger(\N_+ +1) $, which yields $\N_+ a^\dagger_p a^\dagger_{-p}= a_p^\dagger a^\dagger_{-p} (\N_+ +2)$,
together with Cauchy-Schwarz, one obtains 
$$\frac d {ds} e^{-s\B} ( \mathcal{N}_+ + 1)^3 e^{s\B} \lesssim e^{-s\B} ( \mathcal{N}_+ + 1)^3  e^{s\B} ,$$
which yields the result via Gronwall's Lemma. 
\end{proof}

By Lemma \ref{N_+selfbound} 
$$\langle \psi,(\N_+ +1) \psi \rangle = \langle e^{\B}\psi, e^{\B} (\N_+ +1) e^{-\B} e^{\B} \psi \rangle \leq C \langle \psi e^{-\B}  (\N_+ +1) e^{\B} \psi \rangle,$$
which finally allows us to conclude from \eqref{cond} 
that $$ \langle \psi, \N_+ \psi \rangle \leq O(1),$$
which implies complete BEC condensation.
Further 
$$ \langle \psi, \N_0 \psi \rangle = \langle \psi, (\N - \N_+)\psi \rangle = n + O(1),$$
or expressed differently 
$$ \langle \phi_0,\gamma_\psi \phi_0 \rangle = \langle \psi, a^\dagger_0 a_0 \psi \rangle  = n + O(1) .$$

\section{Proof of Lemma \ref{error}}

The proof of Lemma \ref{error} was carried out in detail in \cite{BBCS}. 
The estimates of some terms are tedious, however, straightforward. 
The goal of this section is to outline and streamline the strategy of \cite{BBCS}. 
Let us start with collecting some information about $\vphi_p$. 

\begin{lemma}\label{Lemmavphi} For small enough $\kappa$ one has
\begin{align}\label{infvphi}
& \|\vphi_p\|_\infty  \lesssim \kappa \\\label{2vphi}
& \|\vphi_p \|_2  \lesssim \kappa \\ \label{1vphi}
& \|\vphi_p \|_1  \lesssim \kappa N \\ \label{sumVvphi}
 & \frac 1{N} \sum_p |\hat V(p/N) \vphi_p|   \lesssim \kappa 
\end{align}
\end{lemma}
\begin{proof}
The estimates \eqref{2vphi}, \eqref{infvphi}, \eqref{sumVvphi} are a  consequence of the inequality 
\begin{equation}\label{suppp} 
\sup_{p \in \Lambda^*\setminus \{0\}}|p^2 \vphi_p| \lesssim \kappa.
\end{equation}
To see this, recalling \eqref{scattequ}, we estimate
$$ | \frac {\kappa}{N} \sum_q \hat V((p-q)/N) \vphi_q | \leq \frac {\kappa}{N} \sum_q \frac {|\hat V((p-q)/N)| }{q^2} |q^2 \vphi_q| \lesssim \kappa \sup_q |q^2 \vphi_q| ,$$
with $$ \frac {\kappa}{N} \sum_q \frac {|\hat V((p-q)/N)| }{q^2} = \frac {\kappa}{N^3} \sum_q \frac {|\hat V((p-q)/N)| }{(q/N)^2} \leq C \kappa, $$
which can be seen by treating the latter expression as the Riemann sum of the convolution $|\hat V| \ast 1/p^2$ which is a uniformly bounded function. 
Hence one can bound the absolute value of the left hand side of \eqref{scattequ} from below by 
$\sup_{p}|p^2 \vphi_p|(1- C \kappa)$ which implies \eqref{suppp} using the boundedness of $\hat V$. 
The inequality $|\vphi_p| \lesssim \kappa/p^2$ immediately implies \eqref{2vphi}, \eqref{infvphi}, \eqref{sumVvphi}.
In order to see \eqref{1vphi} we need more decay for large $p$. In configuration space this corresponds to more smoothness. This is usually implied by a bootstrap argument. This can be done here as well. 
Plugging $|\vphi_p| \lesssim \kappa/p^2$  into equation\eqref{scattequ} yields 
$$ |\vphi_p| \lesssim \frac{\kappa}{N^2 (p/N)^2} \left[ \hat V(p/N) + \frac 1 {N^3} \sum_q \hat V((p-q)/N) \frac \kappa{(q/N)^2} \right],$$
which implies 
\begin{equation}\label{boundvphi} |\vphi_p|  \lesssim \frac{\kappa}{N^2 (p/N)^2} G(p/N),
\end{equation}
with $G(p/N) $ at least bounded by
$$ G(p/N) \lesssim 1/(1 + (p/N)^2).$$
Continuing the bootstrap argument allows to improve the fall off properties of $\vphi_p$ further, which however, is not necessary for our purpose. 
The bound \eqref{boundvphi} implies \eqref{1vphi} by considering the sum as Riemann sum. 
\end{proof} 

Let us remark that we perform all our estimates on states $\xi \in \F$, with $$\N\xi = n \xi, \quad n \leq 10 N.$$
Equivalently we will frequently use the operator estimates 
$$ \N \lesssim N, \,\, \N_0 \lesssim N, \,\, \N_+ \lesssim N.$$ 

We start with looking at the terms $\Xi + \E$. The strategy is rather straightforward. 
Whenever the terms inside the bracket of $e^{-s\B}(...)e^{s\B}$ can be estimated by $C (\N_+ + 1)^m$,
then Lemma \ref{N_+selfbound} can be used to bound the total expressions by $ \lesssim \kappa (\N_+ +1)$.

Let us demonstrate this in the case of the first two terms of $\Xi$ in \eqref{CQ2B} as well as $e^{-\B} H_2 e^{\B}$.
For convenience denote 
\begin{multline}
\Xi_1 = -\frac{\kappa} {N^2} \sum_p  \hat{V}(p/N) \vphi_p   \int_0^1 \int_s^1 e^{-t\B} \N_+(2 \N - \N_+ -1 )e^{t\B}\dt \ds \\ +  
 \frac{2\kappa} {N^2}  \sum_p \hat{V}(p/N) \vphi_p \int_0^1 \int_s^1 e^{-t\B}\N_0(\N_0 -1) a^\dagger_p a_pe^{t\B} \dt \ds. 
\end{multline}
Using $\N_+ \leq \N \lesssim N$ and  Lemma \ref{N_+selfbound} we derive
$$ \Xi_1 \lesssim \kappa (\N_+ + 1) \frac 1N \sum_p |\hat V(p/N) \vphi_p | \lesssim \kappa (\N_+ + 1),$$
where in the last step we used \eqref{sumVvphi}. The estimate for $e^{-\B} H_2 e^{\B}$ works in an analogous way. 
Hence we obtain
\begin{equation}
\langle \psi, (\Xi_1 + e^{-\B} H_2 e^{\B}) \psi \rangle \lesssim  \kappa \langle \psi, (\N_+ +1 ) \psi \rangle.
\end{equation}

Next, let us look at the term $\Gamma$. This cannot simply be estimated by $\N_+$. 
We additionally need the interaction $Q_4$. This is no problem as long as the term is not in between
$e^{-s\B} ... e^{s\B}$, due to the fact that $e^{-s\B} Q_4 e^{s\B}$ cannot be dominated by $Q_4$ and $\N_+$, since $e^{-s\B} Q_4 e^{s\B}$ produces an terms of order $N$. %However, in order to recover the potential energy $Q_4$ we have to work in configuration space. 

\begin{remark}
In the following it is important to absorb some error terms by the interaction term $Q_4$. To this aim let us recall first that  
the bosonic Fock space can written as 
$$  \mathcal{F}(\mathcal{H}) =  \mathcal{F}(\{\phi_0\} \oplus \mathcal{H}_\perp) =  \mathcal{F}_0 \otimes  \mathcal{F}_\perp$$
where  $  \mathcal{F}_0$ is the Fock space spanned by the one dimensional space $\{\phi_0\}$ and $\mathcal{F}_\perp = \mathcal{F}(\mathcal{H}_\perp)$, the Fock space built by all states orthogonal to $\phi_0$. In order to see the positivity of $Q_4$ one has to rewrite the term in configuration space. With $$ \check a_x = \sum_p a_p e^{i p \cdot x} , \quad a_p = \int_{\Lambda} \check a_x e^{-ix\cdot p},$$ one can check that  for states $\eta \in \F_\perp$ 
$$ \langle \eta, Q_4 \eta \rangle = \langle \eta,  \iint dx dy   \kappa V_{N}(x-y) \ca^\dagger_x \ca^\dagger_y \ca_y\ca_x 
 \eta \rangle = \iint dx dy   \kappa V_{N}(x-y) \|\ca_y\ca_x \eta  \|^2, $$
where the last term is fundamentally positive for any positive interaction. For that reason it turns out to be convenient to estimate some of the error terms in configuration space.
Hence, whenever we use the interaction $ \iint dx dy   \kappa V_{N}(x-y) \ca^\dagger_x \ca^\dagger_y \ca_y\ca_x $ it has to be remembered that it {\em only} acts on
$\F_\perp$. For sake of convenience we omit the corresponding symbols indicating the restrictions on $\F_\perp$. 
\end{remark}

Lets come back to the term $\Gamma$. 
 Denoting $\Lambda = [-1/2,1/2]^3$ we calculate 
\begin{align*}
&  \sum  \hat{V}_{N}(r) \vphi_p a_{p+r}^\dagger a_q^\dagger a_{-p}^\dagger a_{q+r}\,  a_0 a_0 
\\
&= \sum_{r,p,q} \hat{V}_{N}(r) \vphi_p \int \ca_x^\dagger e^{ix(p+r)} dx \int_\Lambda \ca_y^\dagger e^{iyq} dy \int_\Lambda \ca_z^\dagger e^{-izp} dz \int_\Lambda \ca_w e^{-iw(q+r)} { d}w a_0 a_0 
\\
&= \sum_{p,r} \iiint_{\Lambda^3} dx dy dz \hat{V}_{N}(r) \vphi_p \ca_x^\dagger \ca_y^\dagger \ca_z^\dagger \ca_y e^{ip(x-z)} e^{ir(x-y)} a_0 a_0 
\\
&= \iint_{\Lambda^2} dx dy  V_N(x-y) \ca_x^\dagger \ca_y^\dagger \int \ca_z^\dagger 
\underbrace{\vphi(x-z)}_{=: \vphi_x(z)} 
 dz \ca_y a_0 a_0 
\\
&=\iint_{\Lambda^2} dx dy V_{N}(x-y) \ca_x^\dagger \ca_y^\dagger  \ca^\dagger(\vphi_x) \ca_y a_0 a_0 ,
\end{align*}
where we used from the second to the third line that $\sum_q e^{\ii q \cdot(y-w) }= \delta(y-w)$. 
In terms of expectation values we thus obtain for $\Gamma$ 
\begin{multline}
\frac{\kappa}{N}  |\braket{\xi, \sum   \hat{V}_{N}(r) \vphi(p) a_{p+r}^\dagger a_q^\dagger a_{-p}^\dagger a_{q+r} a_0 a_0 \xi }|
\\
\leq  \frac{1}{N} \iint dx dy \kappa V_{N}(x-y) \|\ca_y \ca_x \xi \| \|\ca^\dagger(\vphi_x)\ca_y a_0 a_0 \xi\|
\\
\leq   \frac{1}{N} \left( \iint dx dy   \kappa V_{N}(x-y) \|\ca_y\ca_x \xi \|^2 \right)^\frac{1}{2} \left(\iint dx dy \kappa V_{N}(x-y) \|\ca^\dagger(\vphi_x)\ca_y a_0 a_0 \xi \|^2 \right)^\frac{1}{2}
\\
\leq \sqrt{2} \braket{\xi,Q_4 \xi}^\frac{1}{2} \|\vphi\|_2 \kappa^{1/2} \|V_N\|^{1/2}_1 \left \|\frac {\N_0}N (\mathcal{N}_+ + 1)^\frac{1}{2}\N_+^{1/2}\xi \right\|
%
%\leq \sqrt{2} \braket{\xi,Q_4 \xi}^\frac{1}{2} \|\vphi\|_2 \left(\iint dx dy \kappa  V_{N}(x-y) \left \|\frac {\N_0}N (\mathcal{N}_+ + 1)^\frac{1}{2}\ca_y\xi \right\|^2 \right)^\frac{1}{2}
%\\
%&= \sqrt{2} \braket{\psi,Q_4 \psi}^\frac{1}{2} \|\vphi\|_2 \|V_{N}\|_1^\frac{1}{2} \|\mathcal{N}_+\psi\|
%\\
%&\leq \delta \braket{\psi,Q_4\psi} + \frac{2}{4\delta} \|\vphi\|_2^2 \|V_{N}\|_1 \|\mathcal{N}_+\psi\|^2
%\\
%\leq \delta\braket{\xi, Q_4\xi} + \kappa^3 C_\delta \braket{\xi, (\mathcal{N}_+ +1)  \xi}.
\end{multline}
By means of Cauchy-Schwarz, and the fact that $\|V_N\|_1 = \|V\|_1/N$, we conclude
that  for any $\delta$ there is a $C_\delta$ such that 
\begin{equation}
\langle \xi, \Gamma \xi \rangle \leq \delta \langle \xi,  Q_4 \xi \rangle + C_\delta \kappa^3 \langle \xi, (\N_+ +1 ) \xi \rangle.
\end{equation}

Next, consider the term 
$$ \int_0^1 \int_0^s e^{-t\B} [\Gamma,B] e^{t\B} \dt \ds.$$
For convenience, we neglect the terms $a_0 a_0/N$, which are bounded by a constant anyway at the end. 
To this aim we first calculate the commutator $[\Gamma,\B]$
\begin{align*}
[\Gamma,\B]=& \frac \kappa{N^2}  \left[\sum \hat{V}_{N}(r)\vphi_p a_{p+r}^\dagger a_q^\dagger a_{-p}^\dagger a_{q+r}a_0a_0 + h.c., \frac{1}{2}\sum \vphi_l (a_l^\dagger a_{-l}^\dagger a_0a_0 - h.c.) \right] 
\\
= &\frac{\kappa}{2N^2 } \sum_{p,q,r,l} \hat{V}_{N}(r)\vphi_p  \vphi_l \left[a_{p+r}^\dagger a_q^\dagger a_{-p}^\dagger a_{q+r} a_0a_0 , a_l^\dagger a_{-l}^\dagger a_0a_0  - a_{-l}a_l a_0^\dagger a_0^\dagger \right] + h.c.
\end{align*}
Evaluating these commutators leads to three types of terms. First 
\begin{equation}\label{Gamma1term}
\frac\kappa{N^2}  \sum_{p,q,r}  \hat{V}_{N}(r) \vphi_p a_{p+r}^\dagger a_q^\dagger a_{-p}^\dagger a_{-q-r}^\dagger \vphi_{q+r} a_0a_0a_0a_0 + h.c.
\end{equation}
second, $-\N_0(\N_0 -1)/N^2$ times the expression
\begin{multline}\label{Gamma2term}
 \kappa \sum_{p,q,r}  \hat{V}_{N}(r) \vphi_p \left( a_{p+r}^\dagger a_q^\dagger a_p a_{q+r} \vphi_p + a_{p+r}^\dagger a_{-p}^\dagger a_{-q} a_{q+r} \vphi_q 
 + a_q^\dagger a_{-p}^\dagger a_{-p-r}a_{q+r}\vphi_{p+r} \right) 
\\
+ \kappa \sum \hat{V}_{N}(r) \vphi_p^2 a_{p+r}^\dagger a_{p+r} +  \kappa \sum \hat{V}_{N}(0)\vphi_p^2 a_q^\dagger a_q + \kappa \sum \hat{V}_{N}(r)\vphi_p \vphi_{p+r}a_p^\dagger a_p.
\end{multline}
The third term stems from the commutator $[a_0a_0,a^\dagger_0a^\dagger_0] = 2(2\N_0 + 1)$, i.e., 
%\begin{align*}
%&= \kappa \sum_{p,q,r,l} \hat{V}_{N}(r)\vphi_p \vphi_l \left(a_{p+r}^\dagger a_q^\dagger a_{-p}^\dagger a_l^\dagger \delta_{l,-q-r} + a_{p+r}^\dagger a_q^\dagger a_l a_{q+r} \right. \delta_{l,p} 
%\\
%& \qquad\left. + a_{p+r}^\dagger a_l a_{-p}^\dagger a_{q+r} \delta_{l,-q} + a_l a_q^\dagger a_{-p}^\dagger a_{q+r} \delta_{l,-p-r} \right) + h.c.
%\\
%& = \kappa \sum_{p,q,r}  \hat{V}_{N}(r) \vphi_p \left(a_{p+r}^\dagger a_q^\dagger a_{-p}^\dagger a_{-q-r}^\dagger \vphi_{q+r} + a_{p+r}^\dagger a_q^\dagger a_p a_{q+r} \vphi_p + a_{p+r}^\dagger a_{-p}^\dagger a_{-q} a_{q+r} \vphi_q \right.
%\\
%&\qquad \left. + a_q^\dagger a_{-p}^\dagger a_{-p-r}a_{q+r}\vphi_{p+r} \right) + h.c.
%\\
%&+ 2\kappa \sum \hat{V}_{N}(r) \vphi_p^2 a_{p+r}^\dagger a_{p+r} + 2 \kappa \sum \hat{V}_{N}(0)\vphi_p^2 a_q^\dagger a_q + 2\kappa \sum \hat{V}_{N}(r)\vphi_p \vphi_{p+r}a_p^\dagger a_p.
%\end{align*}
\begin{equation}\label{Gamma3term}
\frac {\kappa (2\N_0 + 1)} {N^2}  \sum_{p,q,r,k} \hat{V}_{N}(r)\vphi_p \vphi_k a_{p+r}^\dagger a_q^\dagger a_{-p}^\dagger a_{q+r} a_{-k} a_{k}  + \rm{h.c.}
\end{equation}
Recall that all terms have to be sandwiched between $e^{-t\B}... e^{t\B}$. This complicates the estimates whenever it is not possible to bound the terms solely by the 
number operator $\N_+$, but instead  we are forced to use the potential $Q_4$. 
For that reason we postpone the estimation of the terms \eqref{Gamma1term} and \eqref{Gamma3term}.
For the moment we only concentrate on \eqref{Gamma2term}.
Let us start with the quadratic expressions in the last line in \eqref{Gamma2term}. 
Since $\frac {(\N_0 + 1)} {N}  \lesssim 1$, 
the corresponding first two terms in the second line of \eqref{Gamma2term} are simply bounded by 
$$ \frac \kappa{N}  \sum_{r,p} \hat{V}((r-p)/N) \vphi_p^2 a_{r}^\dagger a_{r} +  \frac \kappa{N}  \sum_{p,q} \hat{V}(0)\vphi_p^2 a_q^\dagger a_q \leq C\kappa^3 \frac{\N_+}N,$$
using the $L^2$-bound of $\vphi_p$.
For the third quadratic term we use that $$ \frac \kappa N |\sup_{p} \sum \hat{V}(r/N)\vphi_p \vphi_{p+r} | \leq \frac \kappa N  \| \hat{V} \|_\infty \|\vphi_p\|_1 \|\vphi_p\|_\infty \leq C \kappa^3,$$
with Lemma \ref{Lemmavphi},
such that 
$$ \kappa \sum_{r,p} \hat{V}_{N}(r)\vphi_p \vphi_{p+r}a_p^\dagger a_p \lesssim  \kappa^3 \N_+.$$ 
The analogue estimates hold after performing the integrals $ \int_0^1 \int_0^s e^{-t\B} ...  e^{t\B} \dt \ds$  using Lemma \ref{N_+selfbound}.

Concerning the quartic terms in \eqref{Gamma2term}, the  term where both functions $\vphi_{\ast}$ have the same index cannot be estimated solely by $\N_+$ either, but also needs the interaction $Q_4$. 
The other terms, however,  can   simply be bounded by $\N_+$. Thanks to Lemma \ref{N_+selfbound}  the application of $e^{-\B} ... e^{\B}$ lets the bounds unchanged. 
Let us demonstrate such an estimate on the term including $\vphi_p\vphi_q$. Using Cauchy-Schwarz in $p,q,r$ we obtain
\begin{multline}
 \frac \kappa N \sum_{p,q,r}  \hat{V}(r/N) \vphi_p \vphi_{q} \langle \xi, a_{p+r}^\dagger a_{-p}^\dagger a_{-q} a_{q+r} \xi \rangle \\ \lesssim
 \frac \kappa N \sum_{p,q,r}  |\vphi_q| \| a_{-p} a_{p+r} \xi \| \| a_{-q} a_{q+r} \xi\| |\vphi_p| \lesssim \kappa^3 \frac{\|\N_+ \xi\|^2}{N} \lesssim \kappa^3 \|\N_+^{1/2} \xi \|^2. 
 \end{multline}

\subsection{Remaining terms} 

Finally we collect the remaining terms 
\begin{multline} \label{remterms}
e^{-\B} Q_3 e^{\B} -  \frac {\kappa } {N^2} \sum_{r,p,q}  \hat{V}(r/N)\vphi_p^2   \int_0^1 \int_0^s e^{-t\B} a_{p+r}^\dagger a_q^\dagger a_p a_{q+r} \N_0 (\N_0 +1)e^{t\B} \dt \ds  \\ - \frac{\kappa}{N^3} \sum_{p,q,r}  \hat{V}(r/N) \vphi_{q+r} \vphi_p  \int_0^1 \int_s^1 e^{-t\B}  a_{p+r}^\dagger a_q^\dagger a_{-p}^\dagger a_{-q-r}^\dagger a_0 a_0a_0a_0 e^{t\B} \dt \ds+ h.c. \\ -   \frac{\kappa}{2N^2} \sum_{p,q}  \hat{V}(p/N) \vphi_q \int_0^1 \int_s^1 e^{-t\B}  a^\dagger_{p} a^\dagger_{-p} a_{-q} a_{q} (1 + 2 \N_0)  e^{t\B} \dt \ds\\ + \frac {\kappa} {N^2}  \sum_{p,q,r,k} \hat{V}_{N}(r)\vphi_p \vphi_k \int_0^1 \int_s^1 e^{-t\B} a_{p+r}^\dagger a_q^\dagger a_{-p}^\dagger a_{q+r} a_{-k} a_{k}  (\N_0 + 1) e^{t\B} \dt \ds,
\end{multline}
whose estimates are more elaborate. 
Brennecke and Schlein realized in \cite{BS} that these terms can be expressed via a convergent geometric sum where the bounds of each terms can be classified in a straight forward way. Thereby the convergence is guaranteed by the smallness of $\kappa$.
%The main  difficulty lies in the fact that $e^{-\B} Q_4 e^{\B}$ cannot be controlled simply by $Q_4$ and $\N_+$. 
In the following we present an alternative way of estimating these terms.
The method is essentially the same for all  terms in \eqref{remterms}. We will demonstrate the method on the first and the last two terms. The others work analogously.  
We start with the term $e^{-\B} Q_3 e^{\B}$. To this aim we rewrite it as 
\begin{equation}\label{43}
e^{-\B} Q_3 e^{\B} = \frac {\kappa}{N} \sum_{q,r,q+r} \hat{V}(r/N)\left[e^{-\B}a_{q+r}^\dagger a_{-r}^\dagger e^{\B} e^{-\B}  a_{q}  a_0 e^{\B}  + h.c.  \right].
\end{equation}
Via Duhamel's formula we have 
\begin{equation}\label{Duham}
 e^{-\B}a_{q+r}^\dagger a_{-r}^\dagger e^{\B} = a^\dagger_{q+r} a^\dagger_{-r} + \int_0^1 e^{-s\B} [a^\dagger_{q+r} a^\dagger_{-r}, \B] e^{s\B} ds .
 \end{equation} 
The idea behind this is the simple fact that the corresponding term in \eqref{43} involving $a^\dagger_{q+r} a^\dagger_{-r}$ will be estimated by $Q_4$.
The remaining terms, however, coming from $[a^\dagger_{q+r} a^\dagger_{-r}, \B] $ can be bounded by $\N_+$, which is stable under application of $e^{-\B} .. e^{\B}$. 
% 
%
%The point is now that the commutators with $\B$ produce at least on function $\vphi_{p+q}$ or $\vphi_{-p}$, such that Cauchy-Schwarz can be applied and,
%as shown in \cite{BBCS}, bounded by $\lesssim \kappa (\N_+ +1),$ under assumption $\N \lesssim N$. 
In order to recover $Q_4$ the term 
$$ \frac {\kappa}{N} \sum_{q,r} \hat{V}(r/N) a_{q+r}^\dagger a_{-r}^\dagger e^{-\B}  a_{q}  a_0 e^{\B}$$
has to be estimated in configuration space, where the term reads 
$$
\kappa \int_{\Lambda^2} dx dy  V_N(x-y) \ca^\dagger_x \ca_y^{\dagger} e^{-\B}  \ca_{x}  a_0 e^{\B},
$$
whose expectation value of $\xi$ is bounded by
\begin{multline}
\left( \int_{\Lambda^2} dxdy \kappa V_N(x-y) \| \ca_y \ca_x \xi\|^2\right)^{1/2} \left( \int_{\Lambda^2} dxdy \kappa V_N(x-y) \| e^{-\B} \ca_x a_0 e^{\B} \xi \|^2\right)^{1/2} \\ \lesssim \delta \langle \xi, Q_4 \xi \rangle + \kappa \langle \xi, (\N_+ +1) \xi \rangle.
\end{multline}

The remaining term has the form
$$  \frac {\kappa}{N} \sum_{q,r} \hat{V}(r/N)\int_0^1e^{-s\B}[a_{q+r}^\dagger a_{-r}^\dagger,\B] e^{s\B} e^{-\B}  a_{q}  a_0 e^{\B} ds . $$
Since
$$ [a_{q+r}^\dagger a_{-r}^\dagger,\B] = -\frac 2 N  \left( \vphi_r a^\dagger_{q+r} a_r + \vphi_{r +q} a^\dagger_{-r} a_{-q-r} + \vphi_r \delta_{q,0}\right)a^\dagger_0a^\dagger_0,$$
and the fact that the sum $ \sum_{q,r} $, by assumption, does only include indices different from $0$, only the first two terms need to be estimated. Since they are similar we only consider 
\begin{equation}\label{re1}
  \frac {\kappa}{N^2} \sum_{q,r} \hat{V}(r/N) \vphi_r \int_0^1e^{-s\B}a_{q+r}^\dagger a_{r} a^\dagger_0 a^\dagger_0 e^{s\B} e^{-\B}  a_{q}  a_0 e^{\B} ds . 
  \end{equation}
Using Cauchy-Schwarz for the expectation value of $\xi$ we deduce
$$
|\eqref{re1}| \lesssim \frac{\kappa}{N^2} \sum_{r,q} |\vphi_r| \int_0^1 ds \|a_{q+r} a_0 a_0 e^{s\B} \xi\| \| a_r e^{(s-1) \B} a_q a_0 e^{\B} \xi \| \lesssim \kappa^2 \| (\N_+ +1) \xi \|^2,
$$
where we used $|a_0 a_0| \lesssim \N \lesssim N$, Lemma \ref{N_+selfbound}, and 
\begin{multline}
\sum_{r,q}  \| a_r e^{(s-1) \B} a_q a_0 e^{\B} \xi \|^2 \leq \sum_q \| \N^{1/2} e^{(s-1) \B} a_q a_0 e^{\B} \xi \|^2 = \sum_q \|e^{(s-1) \B}  \N^{1/2} a_q a_0 e^{\B} \xi \|^2 \\ \leq 
\sum_q \| \N^{1/2} \N_0^{1/2} a_q  e^{\B} \xi \|^2 \lesssim N^2 \|(\N_+ + 1)^{1/2}  \xi \|^2.
\end{multline}
Next we look at the second to last term in \eqref{remterms}. 
%, i.e., the term
%$$  \frac{\kappa}{2N^2} \sum_{p,q}  \hat{V}(p/N) \vphi_q \int_0^1 \int_s^1 e^{-t\B}  a^\dagger_{p} a^\dagger_{-p} a_{-q} a_{q} (2 + 2 \N_0)  e^{t\B} \dt \ds.$$  
To this aim, notice 
 that the operator norm of 
$\Phi = \sum_p \vphi_p a_{-p} a_p$ can be estimated by 
\begin{equation}\label{normPhi}
| \Phi| \lesssim \kappa (\N_+ + 1),
\end{equation}
which can be seen by applying Cauchy-Schwarz
$$ |\sum_p \vphi_p a_{-p} a_p| \leq \sum_p ( \kappa a^\dagger_p a_p + (|\vphi_p|^2/\kappa) a_{-p} a_{-p}^\dagger ) \lesssim \kappa (\N_+ +1) .$$

Further we write 
\begin{multline}
\frac{\kappa}{2N^2} \sum_{p,q}  \hat{V}(p/N) \vphi_q \int_0^1 \int_s^1 e^{-t\B}  a^\dagger_{p} a^\dagger_{-p} a_{-q} a_{q} (2 + 2 \N_0)  e^{t\B} \dt \ds \\ =
\frac{\kappa}{2N} \sum_{p}  \hat{V}(p/N) \int_0^1 \int_s^1 e^{-t\B}  a^\dagger_{p} a^\dagger_{-p} e^{t\B} e^{-t\B} \Phi (2 + 2 \N_0)/N e^{t\B} \dt \ds
\end{multline}
where
\begin{equation}\label{adad} e^{-t\B}  a^\dagger_{p} a^\dagger_{-p} e^{t\B}  =  a^\dagger_{p} a^\dagger_{-p} + \int_0^t e^{-s\B} [a^\dagger_{p} a^\dagger_{-p} ,\B] e^{s\B} ds ,
\end{equation}  
and
$$ - [a^\dagger_{p} a^\dagger_{-p} ,\B] =  \frac{a^\dagger_0 a^\dagger_0}{N} \vphi_p( a^\dagger_p a_p + a^\dagger_{-p} a_{-p} + 1 ) .$$
Hence, the expression coming from the second term on the right hand side of \eqref{adad} gives 
\begin{multline}
\frac{\kappa}{N} \sum_{p}  \hat{V}(p/N) \vphi_p \int_0^1 \int_s^1\int_0^t e^{-\tau \B}   \frac{a^\dagger_0 a^\dagger_0}{N} ( a^\dagger_p a_p + 1 )e^{\tau \B}   e^{-t\B} \Phi (1 + 2 \N_0)/N e^{t\B} \dt \ds \lesssim \kappa^3(\N_+ + 1) .
\end{multline} 
The first term on the right hand side of \eqref{adad}, including $ a^\dagger_{p} a^\dagger_{-p}$, is again evaluated by rewriting it in configuration space 
$$ \int_0^1 \int_s^1 \int_{\Lambda^2} dxdy \kappa V_N(x-y) \ca_x^\dagger \ca_y^\dagger e^{-t\B} \Phi (2 + 2 \N_0)/N e^{t\B} \dt \ds,$$
which can be bounded by $$ \delta Q_4 + \kappa^2 (\N_+ +1).$$

Finally, consider the last term in \eqref{remterms}, which we conveniently rewrite as 
$$ \frac {\kappa} {N^2}  \sum_{p,q,r} \hat{V}_{N}(r)\vphi_p \int_0^1 \int_s^1 e^{-t\B} a_{p+r}^\dagger a_q^\dagger e^{t\B}e^{-t\B} a_{-p}^\dagger a_{q+r} \Phi (\N_0 + 1) e^{t\B} \dt \ds,$$ 
using again the notation $\Phi = \sum_k \vphi_k a_{-k} a_k$. Next we apply Duhamel again, similar to \eqref{Duham}, to  $e^{-t\B} a_{p+r}^\dagger a_q^\dagger e^{t\B}$
and obtain two terms, where the first one including $a_{p+r}^\dagger a_q^\dagger$ has to be estimated by $Q_4$. More precisely,
in configuration space that term has the form 
$$  \int_0^1 \int_s^1 \int_{\Lambda^2} dxdy \kappa V_N(x-y)   \ca_{x}^\dagger \ca_y^\dagger e^{-t\B} \frac{\ca^\dagger(\vphi_x) \ca_{y} \Phi (\N_0 + 1)}{N^2} e^{t\B} \dt \ds \lesssim  \delta Q_4 + \kappa^2 (\N_+ +1).$$ 
The second term involving $$\int_0^t e^{-s\B} [a^\dagger_{p+r} a^\dagger_{q}, \B] e^{s\B} ds,$$ 
i.e., 
$$ \frac {\kappa} {N^2}  \sum_{p,q,r} \hat{V}_{N}(r)\vphi_p \int_0^1 \int_s^1 \int_0^t e^{-\tau\B} [a^\dagger_{p+r} a^\dagger_{q}, \B] e^{\tau \B}   e^{-t\B} a_{-p}^\dagger a_{q+r} \Phi (\N_0 + 1) e^{t\B} d\tau  \dt \ds,$$

is again estimated by $\kappa^3(\N_+ + 1)$.

\section{Proof of Lemma \ref{scat}}

Using \eqref{invfour} we can write 
\begin{multline}
 \frac {1}{N} \sum_{p} \hat V(p/N) \vphi_p = N^2 \frac 1{N^3} \sum_{p\neq 0} \hat V(p/N) \vphi_p= N^2 \langle P_0^\perp \vphi, P_0^{\perp} V  \rangle_N 
 \\ = - \left \langle   V, P_0^\perp
\frac 1{P_0^\perp(-\Delta + \frac \kappa 2 V )P_0^\perp} P_0^\perp \frac \kappa 2 V \right \rangle_N, \qquad \qquad 
\end{multline} 
where we used \eqref{deffi} to obtain the last equality. 

This implies now for $a_N$
\begin{multline}
4\pi a_N = \frac{\kappa \hat{V}(0) }{2} + \frac{1 }{N} \sum_p \frac \kappa 2\hat{V}(p/N) \vphi_p  \\ = 
\frac \kappa 2 \int_{[-\frac N2 ,\frac N2]^3} V(x) -  \left \langle \frac \kappa 2 V, P_0^\perp
\frac 1{P_0^\perp(-\Delta + \frac \kappa 2 V )P_0^\perp} P_0^\perp \frac \kappa 2 V \right \rangle_N.
% \\
%= \frac{\N_0^2}N \left \langle \sqrt{v} , \frac 1 {1 + \sqrt{v} \frac 1{\epsilon_\mu} \sqrt{v}} \sqrt{v} \right \rangle,
\end{multline}
On a formal level this converges for $N\to \infty$ as
\begin{multline}
\frac \kappa 2 \int_{[-\frac N2 ,\frac N2]^3} V(x) -  \left \langle \frac \kappa 2 V, P_0^\perp
\frac 1{P_0^\perp(-\Delta + \frac \kappa 2 V )P_0^\perp} P_0^\perp \frac \kappa 2 V \right \rangle_N  \to_{N\to\infty}  \\
\frac \kappa 2 \int_{\R^3}V(x) -  \left \langle \frac \kappa 2 V, 
\frac 1{(-\Delta + \frac \kappa 2 V )}  \frac \kappa 2 V \right \rangle =  \left \langle \sqrt{v} , \frac 1 {1 + \sqrt{v} \frac 1{-\Delta} \sqrt{v}} \sqrt{v} \right \rangle,
\end{multline}
with $v = \frac \kappa 2 V$. The right hand side is $4\pi a$, with $a$ being the scattering length. 
In order to obtain the bound $|a - a_N | \leq O(1/N)$, 
lets denote $$ \1_N = \sum_{p\neq 0} | e_p\rangle \langle e_p |, \qquad | e_p\rangle = e^{-i \frac pN \cdot x}/N^{3/2},$$
such that $\1_N +|e_0\rangle \langle e_0 | $ is the identity on $L^2([-N/2,N/2]^3)$. Then we can write 
$$ \left \langle v , P_0^\perp
\frac 1{P_0^\perp(-\Delta + v)P_0^\perp} P_0^\perp v\right \rangle_N = \left \langle v, \1_N \frac 1{-\Delta + v} \1_N v \right \rangle = \left \langle v, (\1 - \delta_N) \frac 1{-\Delta + v} (\1-\delta_N) v \right \rangle ,$$ 
with $\delta_N = \1 - \1_N$, where $\1$ denotes the identity on $L^2(\R^3)$. Notice, we implicitly assume that the application of $\1_N$ means that one only integrates over $[-N/2,N/2]^3$. We also assume that $N$ is large enough such that the support of $v$ is in $[-N/2,N/2]^3$. 
Hence, we can write the difference of $a - a_N$ as
\begin{multline}\label{diffaaN}
 a- a_N = \left \langle v ,\1_N \frac 1{-\Delta + v} \1_N v \right \rangle - \left \langle v, \frac 1{-\Delta + v} v \right \rangle \\ = - 2 \Re\left  \langle v ,\delta_N \frac 1{-\Delta + v}  v \right \rangle +  \left \langle v ,\delta_N \frac 1{-\Delta + v} \delta_N v\right \rangle. 
 \end{multline}  
Observe that $\vphi = - \frac 1{-\Delta + v} v $ is the solution of the scattering equation on the whole space, which is smooth function with falloff $1/|x|$ in configuration space, due to the properties of $V$. 
The first term of the right hand side of \eqref{diffaaN} is  
$$\left \langle v ,\delta_N \frac 1{-\Delta + v}  v \right \rangle = \int_{\R^3} \hat V(p) \hat \vphi(p) dp -  \frac 1{N^3} \sum_p \hat \vphi(p/N) \hat v(p/N),$$ 
which is the difference of the Riemann sum and its integral. A second order Taylor expanding of $\hat \vphi(p)$ shows that this error is of order $1/N$ 
due to the $1/p^2$ behavior for small $p$. Notice that a quick first order expansion gives an error of $\log N /N$. 

For the second term in \eqref{diffaaN}, observe 
\begin{multline} 
\left \langle v, \delta_N \frac 1{-\Delta + v} \delta_N v \right \rangle \leq \left \langle v, \delta_N \frac 1{-\Delta} \delta_N v \right \rangle = \\ = \int_{\R^3} \frac{|\hat v(p)|^2}{p^2} dp - \frac 1{N^3} \sum_p  \frac{|\hat v(p/N)|^2}{(p/N)^2} = O(1/N),
\end{multline} 
which is again the difference of a specific integral and its Riemann approximation.

\subsection*{Acknowledgements}

I am deeply indebted to Benjamin Schlein for patiently explaining strategies and estimates of their work (BBCS) and several bounds used in the present paper. 
I am further thankful to Jan Philip Solovej for explanation of basic ideas in \cite{BSo}. Further I want to thank Phan Thanh Nam for many helpful discussions, as well as Peter M\"uller for valuable explanations. 
\\

\medskip 
The data that supports the findings of this study are available within the article


\begin{thebibliography}{55}



\bibitem{ABS}
A.~Adhikari, C.~Brennecke, B.~Schlein. 
Bose-Einstein Condensation Beyond the Gross-Pitaevskii Regime.
{\it Annales Henri Poincar\'e } {\bf 22}, (2021) 1163--1233

%\bibitem{BDS}
%N.~{Benedikter}, G.~{de Oliveira} and B.~{Schlein}. {Quantitative derivation of the Gross-Pitaevskii equation}. \emph{Comm. Pure Appl. Math.} (2014).

%\bibitem{BFKT}
%T. Balaban, J. Feldman, H. Kn\"orrer, E. Trubowitz. Complex Bosonic Many-Body Models:
%Overview of the Small Field Parabolic Flow. {\it Ann. Henri Poincar\'e } {\bf 18} (2017), 2873--2903.

\bibitem{BBCS0}
C. Boccato, C. Brennecke, S. Cenatiempo, B. Schlein. Complete Bose-Einstein condensation in the Gross-Pitaevskii regime. {\it Comm. Math. Phys.} {\bf 359} (2018), no. 3, 975--1026.

%\bibitem{BBCS1}
%C. Boccato, C. Brennecke, S. Cenatiempo, B. Schlein. The excitation spectrum of Bose gases interacting through singular potentials.  Preprint arXiv:1704.04819. To appear on {\it J. Eur. Math. Soc.} 

\bibitem{BBCS}
C. Boccato, C. Brennecke, S. Cenatiempo, B. Schlein. Bogoliubov Theory in the Gross-Pitaevskii limit. \emph{Acta Mathematica} \textbf{222} (2019), no. 2, 219--335.

\bibitem{BBCS1}
C. Boccato, C. Brennecke, S. Cenatiempo, B. Schlein. Optimal Rate for Bose-Einstein Condensation in the Gross-Pitaevskii Regime. \emph{Comm. Math. Phys.} (2019), doi:10.1007/s00220-019-03555-9. Preprint arXiv:1812.03086.

\bibitem{bog}
N. N. Bogoliubov. On the theory of superfluidity.
{\it Izv. Akad. Nauk. USSR} {\bf 11} (1947), 77. Engl. Transl. {\it J. Phys. (USSR)} {\bf 11} (1947), 23. 

%\bibitem{BCS}
%C. Brennecke, M. Caporaletti, B. Schlein. Excitation spectrum for Bose gases beyond the Gross-Pitaevskii regime. In preparation. 
%
\bibitem{BS}
C. Brennecke, B. Schlein. Gross-Pitaevskii dynamics for Bose-Einstein condensates. \emph{Analysis \& PDE} \textbf{12} (2019), no. 6, 1513--1596.

\bibitem{BSo}
%B. Brietzke. On the Second Order Correction to the Ground State Energy of the Dilute Bose Gas. PhD %thesis (2017).
B. Brietzke, J.P. Solovej. The Second Order Correction to the Ground State Energy of the Dilute Bose Gas. {\it Annales Henri Poincar\'e} {\bf 21} (2020), 571--626

\bibitem{BFS}
%%B. Brietzke. On the Second Order Correction to the Ground State Energy of the Dilute Bose Gas. PhD %thesis (2017).
B. Brietzke, S. Fournais, J.P. Solovej. A simple second order lower bound to the energy of the dilute Bose gases. {\it Comm. Math. Phys} {\bf 376} (2020), 321--351. 
%


%\bibitem{DN}
%J.~ Derezi\'nski, M.~Napi\'orkowski. Excitation Spectrum of Interacting Bosons in the Mean-Field Infinite-%Volume Limit. {\it Annales Henri Poincar\'e} {\bf 
%15} (2014), 2409-2439. 

\bibitem{Dy}
F.J. Dyson. Ground-State Energy of a Hard-Sphere Gas. {\it Phys. Rev.} {\bf 106} (1957), 20--26.


%\bibitem{ESY0} 
%L.~{Erd\H{o}s}, B.~{Schlein} and H.-T.~{Yau}.
%Derivation of the Gross-Pitaevskii hierarchy for the dynamics of Bose-Einstein condensate. {\it Comm. Pure  Appl. Math.} {\bf 59} (2006), no. 12, 1659--1741.

%\bibitem{ESY}
%L. Erd\H os, B. Schlein, H.-T. Yau. Ground-state energy of a low-density Bose gas: a second order upper %bound. {\it Phys. Rev. A} {\bf 78} (2008), 053627.
%

%\bibitem{ESY2}
%L.~{Erd{\H{o}}s}, B.~{Schlein}, H.-T.~{Yau}.
%\newblock Derivation of the {G}ross-{P}itaevskii equation for the dynamics of
%  {B}ose-{E}instein condensate,
%\newblock \emph{Ann. of Math.} \textbf{172} (2010), no. 1, 291--370.


%\bibitem{EY}
%L.~{Erd{\H{o}}s} and H.-T.~{Yau}. {Derivation of the
%  nonlinear {S}chr\"{o}dinger equation from a many-body {C}oulomb system}. \emph{Adv.
%  Theor. Math. Phys.} \textbf{5} (2001),  no. 6, 1169--1205.

\bibitem{F} S. Fournais. Length scales for BEC in the dilute Bose gas. arXiv:2011.00309

\bibitem{FS}
S. Fournais, J.P. Solovej. The energy of dilute Bose gases. {\it Annals of Mathematics}, {\bf 192} (2020)


\bibitem{HS} C. Hainzl, R. Seiringer. The BCS critical temperature for potentials with negative scattering length. {\it Lett. Math. Phys.}  {\bf 84} (2008)2-3, 99-107 

%\bibitem{GiuS}
%A. Giuliani, R. Seiringer. The ground state energy of the weakly interacting Bose gas at high density. {\it J. %Stat. Phys.} {\bf 135} (2009), 915.
%
%\bibitem{GS}
%P. Grech, R. Seiringer. The excitation spectrum for weakly interacting bosons in a trap. {\it Comm. Math. %Phys.} {\bf 322} (2013), no. 2, 559-591.
%
%\bibitem{LNS} M.~{Lewin}, P.~T.~{Nam} and B.~{Schlein}. 
%{Fluctuations around Hartree states in the mean-field regime}. %Preprint arXiv:1307.0665.
%
%\bibitem{Lan}
%L.D. Landau. Theory of the superfluidity of Helium II.
%{\it Phys. Rev.} {\bf 60} (1941), 356Ð-358.
%
%\bibitem{LNR1} M. Lewin, P.~T.~Nam, N. Rougerie. Derivation of Hartree's theory for generic
%mean-field Bose gases. {\it Adv. Math.} {\bf 254} (2014), pp. 570-621. 

%\bibitem{LNR2} M. Lewin, P.~T.~Nam, N. Rougerie. The mean-field approximation and the 
%non-linear Schr\"odinger functional for trapped  {B}ose gases. {\it Trans. Amer. Math. Soc.} {\bf 368} 
%(2016), no. 9, 6131-6157. 

%\bibitem{HY} K. Huang, C. N. Yang. Quantum-Mechanical Many-Body Problem with Hard-Sphere Interaction. \emph{Phys. Rev.} \textbf{105}, 3, (1957), 767 --757.
%
%\bibitem{LHY} T. D. Lee, K. Huang, C. N. Yang. Eigenvalues and Eigenfunctions of a Bose System of Hard Spheres and Its Low-Temperature Properties. \emph{Phys. Rev.} \textbf{106} (1957), 6, 1135--1145.
%
%\bibitem{LY} T. D. Lee, C. N. Yang. Many-Body Problem in Quantum Mechanics and Quantum Statistical Mechanics. \emph{Phys. Rev.} \textbf{105} (1957), 1119 -- 1120.
%
%\bibitem{LNSS} M.~Lewin, P.~T.~{Nam}, S.~{Serfaty}, J.P. {Solovej}. Bogoliubov spectrum of interacting Bose gases.  \newblock \emph{Comm. Pure Appl. Math.} \textbf{68} (2014), 3, 413 -- 471. 

\bibitem{LS1}
E.~H.~Lieb and R.~Seiringer.
\newblock Proof of {B}ose-{E}instein condensation for dilute trapped gases.
\newblock \emph{Phys. Rev. Lett.} \textbf{88} (2002), 170409.

%\bibitem{LS2}
%E.~H.~Lieb and R.~Seiringer.
%\newblock Derivation of the Gross-Pitaevskii equation for rotating Bose gases.
%\newblock \emph{Comm. Math. Phys. } \textbf{264}:2 (2006), 505-537.
%
\bibitem{LSSY}
E.~H.~Lieb, R.~Seiringer, J. P. Solovej and J.~Yngvason. \textit{
The Mathematics of the Bose Gas and its Condensation}. Series: Oberwolfach Seminars. 
Birkh\"auser Verlag, 2005.


\bibitem{LSY}
E.~H.~Lieb, R.~Seiringer, and J.~Yngvason.
\newblock Bosons in a trap: A rigorous derivation of the {G}ross-{P}itaevskii
  energy functional. \newblock \emph{Phys. Rev. A} \textbf{61} (2000), 043602.


\bibitem{LS2}
E.~H.~Lieb and R.~Seiringer.
\newblock Derivation of the Gross-Pitaevskii Equation for Rotating Bose Gases
\newblock \emph{Comm. Math. Phys.} \textbf{264} (2006), 505--537.


%\bibitem{LSo}
%E.~H.~Lieb, J. P. Solovej. Ground state energy of the one-component charged Bose gas. {\it Comm. Math. Phys.} {\bf 217} (2001), 127--163. Errata: {\it Comm. Math. Phys.} {\bf 225} (2002), 219-221.

%\bibitem{LSo2}
%E.~H.~Lieb, J. P. Solovej. Ground state energy of the two-component charged Bose gas. {\it Comm. Math. %Phys.} {\bf 252} (2004), 485--534.

\bibitem{LY} 
E.~H.~Lieb, J.~Yngvason. Ground State Energy of the low density Bose Gas. {\it Phys. Rev. Lett.} {\bf 80} %(1998), 2504Ð2507. 

\bibitem{NNRT} 
P.~T.~{Nam}, M. Napi{\'o}rkowski, J. Ricaud, A. Triay. Optimal rate of condensation for trapped bosons in the Gross--Pitaevskii regime. Preprint arXiv:2001.04364. 

\bibitem{NRS}
P.~T.~{Nam}, N. ~{Rougerie}, R.~Seiringer. Ground states of large bosonic systems: The Gross-Pitaevskii limit revisited. 
\emph{Analysis and PDE.} {\bf 9} (2016), no. 2, 459--485

%\bibitem{NRS1}
%M. Napi{\'o}rkowski, R. Reuvers, J. P. Solovej. The Bogoliubov free energy functional I. Existence of %minimizers and phase diagrams. Preprint arxiv:1511.05935.
%
%\bibitem{NRS2}
%M. Napi{\'o}rkowski, R. Reuvers, J. P. Solovej. The Bogoliubov free energy functional II. The dilute limit.  %Preprint arxiv:1511.05953.
%
%
%
%\bibitem{P1}
%A. Pizzo. Bose particles in a box I. A convergent expansion of the ground state of a three-modes %Bogoliubov Hamiltonian in the mean field limiting regime. Preprint arxiv:1511.07022.
%
%\bibitem{P2}
%A. Pizzo. Bose particles in a box II. A convergent expansion of the ground state of the Bogoliubov %Hamiltonian in the mean field limiting regime. Preprint arxiv:1511.07025.
%
%\bibitem{P3}
%A. Pizzo. Bose particles in a box III. A convergent expansion of the ground state of the Hamiltonian in 
%the mean field limiting regime. Preprint arxiv:1511.07026.
%
%
%\bibitem{Sei}
%R. Seiringer. The Excitation Spectrum for Weakly Interacting Bosons. {\it Comm. Math. Phys.} {\bf 
%306} (2011), 565Ð-578. 

%\bibitem{So}
%J. P. Solovej. Upper bounds to the ground state energies of the one- and two-component charged Bose %gases. {\it Comm. Math. Phys.} {\bf 266} (2006), no. 3, 797-818.
%
%\bibitem{YY}
%H.-T. Yau, J. Yin. The second order upper bound for the ground state energy of a Bose gas. {\it J. Stat. %Phys.} {\bf 136} (2009), no. 3, 453--503.

\end{thebibliography}
\end{document}